\documentclass[11pt,letterpaper]{article}

\usepackage[noend]{algorithmic}

\usepackage[margin=1in]{geometry}
\usepackage{lmodern}
\usepackage[T1]{fontenc}
\usepackage[utf8]{inputenc}
\usepackage[tbtags]{amsmath}
\allowdisplaybreaks
\usepackage{bbm}
\usepackage{amssymb,amsthm}
\usepackage{hyperref}
\usepackage[svgnames]{xcolor}
\definecolor{links}{RGB}{30, 85, 255}
\definecolor{cites}{RGB}{0, 0, 139}
\definecolor{urls}{RGB}{255, 116, 20}
\usepackage[compress]{natbib}
\usepackage{doi}
\usepackage{tikz} 
\usepackage{pgf,pgfplots}
\pgfplotsset{compat=1.14}
\usepgfplotslibrary{fillbetween}
\usepackage{hyphenat}
\usepackage[font=small]{caption}
\usepackage{subcaption}
\usepackage{appendix}
\usepackage{todonotes}

\usepackage[ruled,noend,noline]{algorithm2e}

\SetCommentSty{mycommfont}

\usepackage[capitalise,nameinlink]{cleveref}
\providecommand{\DontPrintSemicolon}{\dontprintsemicolon}
\usepackage{dsfont}
\usepackage{microtype,bm}
\usepackage{booktabs}
\usepackage{natbib}
\usepackage{colortbl}
\usepackage{xcolor}
\usepackage{setspace}

\newtheorem{theorem}{Theorem}

\theoremstyle{definition}

\newtheorem{claim}{Claim}

\DeclareMathOperator*{\argmax}{arg\,max}

\newcommand{\E}[1]{\mathbb E \left[ #1 \right]}

\newcommand{\M}{M} 

\newcommand{\Ms}{M_S} 

\newcommand{\Eprime}{E'} 

\author{
	Paul D{\"u}tting\thanks{Google Research, Z\"urich, Switzerland. Email:   
\texttt{\href{mailto:duetting@google.com }{duetting}@google.com}}
	\and
	Federico Fusco\thanks{Department of Computer, Control, and Management 
Engineering ``Antonio Ruberti'', Sapienza University of Rome, 
Italy. Email:   
\texttt{\{\href{mailto:fuscof@diag.uniroma1.it}{fuscof},
        \href{mailto:lazos@diag.uniroma1.it}{lazos},
		\href{mailto:leonardi@diag.uniroma1.it}{leonardi},
        \href{mailto:rebeccar@diag.uniroma1.it}{rebeccar}\}@diag.uniroma1.it}
}
\and
Philip Lazos{$^\ddag$}
\and
Stefano Leonardi{$^\ddag$}
\and
Rebecca Reiffenh{\"a}user{$^\ddag$}
}

\title{
Prophet Inequalities for Matching with a Single Sample
\thanks{
This work was supported by the ERC Advanced 
Grant 788893 AMDROMA ``Algorithmic and Mechanism Design Research in 
Online Markets'', the MIUR PRIN project ALGADIMAR ``Algorithms, Games, 
and Digital Markets''.}
}

\date{}

\makeatletter
\let\original@algocf@latexcaption\algocf@latexcaption
\long\def\algocf@latexcaption#1[#2]{%
  \@ifundefined{NR@gettitle}{%
    \def\@currentlabelname{#2}%
  }{%
    \NR@gettitle{#2}%
  }%
  \original@algocf@latexcaption{#1}[{#2}]%
}
\makeatother

\begin{document}
\maketitle
\begin{abstract}
We consider the prophet inequality problem for (not necessarily bipartite) matching problems with independent edge values, under both edge arrivals and vertex arrivals.
We show constant-factor prophet inequalities for the case where the online algorithm has only limited access to the value distributions through samples.
First, we give a $16$-approximate prophet inequality for matching in general graphs under edge arrivals that uses only a single sample from each value distribution as prior information. Then, for bipartite matching and (one-sided) vertex arrivals, we show an improved bound of $8$ that also uses just a single sample from each distribution. Finally, we show how to turn our $16$-approximate single-sample prophet inequality into a truthful single-sample mechanism for online bipartite matching with vertex arrivals.
\end{abstract}

\section{Introduction}
We study the  maximum-weight matching problem, i.e., maximizing the cumulative weight of an edge set $M$ in a graph such that to each vertex, at most one edge from $M$ is incident. 
This problem class comes up naturally in a wide range of applications, like matchmaking, two-sided markets, assigning jobs to workers, kidney exchange, housing markets, and the like. In many of these applications, the problem is \emph{online} in that the input is not fully known from the outset. Rather it is revealed over time, and decisions have to be made immediately and irrevocably. The online matching problem comes in many different flavors, for example one can study it for restricted graph classes like bipartite graphs (e.g., jobs/workers), under different online models (e.g., a fixed set of houses for which potential buyers become known only over time), and different assumptions regarding the arrival model (random/adversarial/with expiration times etc.).

\paragraph{Prophet Inequalities.}
We concentrate here on the \emph{prophet} model \citep{KrengelS77,KrengelS78,SamuelCahn84}, in which elements arrive in adversarial order and weights 
are drawn independently from known distributions. In the original definition, the adversary fixes the order and distributions up front.
The goal is to show a \emph{prophet inequality}, i.e. an online algorithm such that the expected weight of the computed solution is at least a $1/\alpha$-fraction of the expected weight of the optimal solution. 

The original work of \cite{KrengelS77,KrengelS78} and \cite{SamuelCahn84} has established a tight $2$-approximate prophet inequality for the problem of selecting one value out of a sequence of $n$ values drawn independently from not necessarily identical distributions. Motivated, in part, by a connection to posted-price mechanisms \citep{HajiaghayiKS07,ChawlaHMS10}, recent work has established prophet inequalities for a broad range of combinatorial problems, including matroids and polymatroids (e.g. \cite{Alaei11,KleinbergW12,DuttingK15,FeldmanSZ16}), bipartite and non-bipartite matching (e.g. \cite{ChawlaHMS10,KleinbergW12,FeldmanGL15,GravinW19,EzraFGT20}), and combinatorial auctions (e.g. \cite{FeldmanGL15,DuttingFKL20,DuttingKL20}).

Specifically, for the matching problem constant-factor prophet inequalities have been established for both (one-sided) vertex arrivals in a bipartite graph \citep{FeldmanGL15},  and edge arrivals in bipartite and non-bipartite graphs  \citep{ChawlaHMS10,KleinbergW12,GravinW19,EzraFGT20}.

\paragraph{The Single-Sample Paradigm.}
For many applications the assumption that the distributions are known exactly, as in the prophet model, is a rather strong one. It is natural to ask whether constant-factor prophet inequalities are also achievable with limited information about the distributions. 
A natural approach in this context (pioneered by \cite{AzarKW14}) is to assume that the online algorithm has access to a limited number of \emph{samples} from the underlying distributions, e.g., from historical data. 
Arguably, a minimal assumption under this approach is that the online algorithm has access to just a \emph{single} sample from each distribution. 

Surprisingly, even under this very restrictive assumption, strong positive results are possible. Importantly, for the classic single-choice prophet inequality one can obtain a factor $2$-approximate prophet inequality with just a single sample from each distribution \citep{RubinsteinWW20}, matching the best possible guarantee that can be given with full knowledge of the distributions.

The same problem with identical distributions was considered by \cite{CorreaDFS19}, who showed that with one sample from each distribution it is possible to achieve a $e/(e-1)$ guarantee. This bound was subsequently improved by \cite{KaplanNR20}, \cite{CorreaCES20}, and \cite{CorreaDFSZ21} with a tight bound appearing in a recent working paper by \cite{abs-2011-06516}. Together these results show that there is a small constant-factor gap between the best-possible guarantees with a single sample, and the optimal bound with full knowledge of the distributions shown in \cite{CorreaFHOV17}.

\citet{AzarKW14} give a general reduction showing that a certain type of secretary algorithms, so-called order-oblivious secretary algorithms, implies the existence of a single-sample prophet inequality with the same approximation guarantee. They use this reduction to derive single-sample prophet inequalities for a range of combinatorial problems.

\subsection{Results and Techniques}

We prove novel and improved constant-factor single-sample prophet inequalities for matching problems via a direct analysis of simple greedy-based algorithms.

At the heart of our results lies a relatively simple fact.
Recall that the greedy algorithm for maximum-weight matching orders all edges of a graph decreasingly by weight. Then, it traverses them in this order and whenever there is no conflict, the currently considered edge is added. It is not a new observation that this well-known and simple $2$-approximation exhibits properties close to those of an online algorithm---actually, when the order of edge's arrival is decreasing by weight, it does behave exactly like one. At the same time, the most natural idea when designing sample-based prophet inequalities is to accept a value during the online algorithm if it beats a threshold derived from the sample(s), e.g., the largest one. We combine the strategy of sample-thresholds with an analysis of the greedy algorithm on the set of \emph{all} values involved, those from the sample as well as the actual problem instance. 

Closest to our work, a similar approach of exploiting the fixed greedy order was previously applied by \cite{KorulaP09} in the context of secretary algorithms, i.e., in the absence of prior knowledge but with the additional assumption of random-order arrivals (see \cref{sec:discussion} for a more detailed discussion). 

We apply the above idea to prophet inequalities: somewhat surprisingly, this allows one to tie the run of the (online, adversarial-order) prophet algorithm closely to that of (offline, fixed-order) greedy, in the sense that the thresholds correspond to skipping certain values in the greedy algorithm. 

First, we analyze the most general setting of max-weight matching, where $G$ is an arbitrary (not necessarily bipartite) graph and the edges arrive one-by-one in adversarial order. We provide a single-sample prophet inequality that computes a matching of expected value at least a $\tfrac{1}{16}$-fraction of the expected value of the maximum weight matching (Theorem~\ref{thm:edgeGen}).

Second, with considerable adjustments in the analysis and a slightly different algorithm, we derive a tighter bound of $8$ for the special case of a bipartite graph and one-sided vertex arrival (Theorem~\ref{thm:vertexBip}).

Finally, since the bipartite graph can be interpreted as dividing the vertices into items and buyers, for this setting we provide also an incentive-compatible mechanism. Because our $8$-approximation appears to lose its approximation guarantee when naturally adjusted to be incentive-compatible, we employ instead a generalized view on the analysis for the edge-arrival case, accommodating also the behavior of the selfish buyers. This again results in a $16$-approximation (Theorem~\ref{thm:truthfulBip}).

\subsection{Related Work}\label{sec:discussion}

Our work is part of a broader literature in algorithm and mechanism design that seeks to obtain near-optimal algorithms 
with access to a single sample from the underlying distributions (e.g.,  \cite{AzarKW14, DhangwatnotaiRY15, Duetting20}). 

The most extensive, previous study of single-sample prophet inequalities is given by \cite{AzarKW14}. Based on the observation that the existence of an \emph{order-oblivious} secretary algorithm implies the existence of a single-sample prophet inequality with the same ratio, they provide constant approximations for a number of matroids. They also use a variation of their reduction argument to obtain a $6.75$-approximation for bipartite matching environments with edge arrivals. For this result they need to assume the vertex degree is limited by a constant $d$, and they require $O(d^2)$ samples, i.e., a constant number of samples per edge rather than just one.

The aforementioned results of \cite{KorulaP09} together with the insights from \cite{AzarKW14} can actually be employed to derive a (truly) constant bound for maximum-weight bipartite matching:
\cite{KorulaP09} present an $8$-approximate secretary algorithm for this problem with one-sided vertex arrival. 
Although they do not specifically aim for it, their algorithm and a version of their analysis is in fact order-oblivious. This, together with the reduction of \cite{AzarKW14} implies the existence of a 13.5-approximate, single-sample prophet inequality for the special case of bipartite graphs where one side is known offline.
This result, like the result for bipartite matching given by \cite{AzarKW14}, relies on the relation of single-sample prophet inequalities to order-oblivious secretary algorithms. 
In our work, we deviate from this traditional path. While we have the use of a greedy pricing in common with \cite{KorulaP09}, our algorithms are directly designed to work in the prophet model, avoiding the inherently wasteful detour over order-oblivious secretary algorithms.

In concurrent and independent work, \cite{CaramanisEtAl21} very recently provide prophet inequalities using an approach not unsimilar to ours. They, too, present a constant-factor approximation for matching in general graphs with edge arrivals, where they show a $32$-approximation single-sample prophet inequality. 
Aiming in a different direction than ours, they then 
provide improved results for several special matroids, among them transversal and truncated partition matroids.
We instead focus on the important case of bipartite matchings, and the problem of incentive compatibility.
Together, our results demonstrate the versatility and power of a direct, greedy-order based approach for the design of single-sample prophet inequalities.

In follow-up work, \citet{Kaplan21} also consider online weighted matching problems with samples. They show an improved bound of $5.83$ for the single-sample prophet inequality problem for bipartite matching with one-sided vertex arrivals, and an improved bound of $13.5$ for matching in general graphs with edge arrivals.

\section{Preliminaries}   

We have a weighted graph $G = (V,E)$, where each edge $e$'s weight $r_e$ is sampled independently from an unknown distribution $D_e$. The only information about these distributions available to the algorithm beforehand is a single sample $s_e$ drawn from $D_e$ for each edge. 
$S$ and $R$ are both drawn independently from the product distributions on the $D_e$ and denote the set of the realized weights of the samples ($S$) and of the actual weights ($R$), respectively. For the sake of simplicity, in the following we often refer to the $R$ graph, which is simply graph $G$ whose edge weights are given by $R$. The $S$ graph is defined analogously but using as weights the values in $S$. We consider the following two models:

\paragraph{General graph and edge arrivals.}
All the vertices $V$ are known in advance, but the edges arrive one-by-one in adversarial order. Before the online procedure starts, the algorithm has access to the sampled edge weights $S$. In the online phase, edges arrive in some (adversarial) order. Upon arrival of the generic edge $e$, the corresponding weight $r_e$, which is drawn from $D_e$, is revealed and the decision whether to add it the solution is made immediately and irrevocably. 

\paragraph{Bipartite graph with vertex arrivals.} The vertices are divided into two sets: buyers $I$ (left vertex set of $G$), and items $J$ (right vertices). 
The $J$ side of the graph (the items) is fixed and available from the beginning together with the realized $S$ graph. On the other hand, the buyers $I$ appear in adversarial online fashion. Each time buyer $i$ arrives, the algorithm either immediately and irrevocable match her with any item $J$ that is still available, or let her leave forever unmatched. Contrary to our edge arrival model, the weights of \emph{all} incident edges of buyer $i$ are revealed at once.
In this section, we also consider the issue of incentive-compatibility or truthfulness if the buyers act as selfish agents. A mechanism (assignment rule together with a pricing rule) is called \emph{IC} or \emph{truthful} if for every agent, her utility (value of item assigned minus price paid) is maximized by reporting her true valuation for all items on sale, i.e., for all adjacent edges in the bipartite graph.

\paragraph{Competitive analysis.} Our objective is to maximize the weight of the output matching $M$, i.e. to find a one-on-one assignment that maximizes the sum of weights $w(M)=\sum_{e\in M} r_e$. Given the online nature of the problem, we follow the literature on prophet inequalities and evaluate the performance of our algorithms in terms of the \emph{competitive ratio}, i.e. the ratio between the expected weight of an offline maximum weight matching and the expected weight of the online matching produced by the algorithm. Clearly, this notion of performance depends on the order in which the elements arrive, so we give bounds that hold for all possible (adversarial) orders. All bounds that we obtain are related to the expected value of the matching computed on the sample values. The bounds  hold also for an {\em adaptive adversary} that knows in advance the whole sequence of values that will be realized, and it will release the next item in the sequence while knowing all the decisions previously made by the algorithm.

\begin{algorithm}[t] 
\DontPrintSemicolon
\caption{Prophet Matching}
    Set $\Eprime=\emptyset$, $\M= \emptyset$\;
    Compute the greedy matching $\Ms$ on the graph with edge weights according to sample $S$\;
    \For{all $e = (u,v) \in \Ms$}{
    Set $p_{u}=p_{v}=s_e$.}
    \For{all $u \in V$, $v$ unmatched in $\Ms$}{
    Set $p_u=0$}
    \For{each arriving $e = (u,v) \in E$}{
        \If{$r_e > \max\{p_u,p_v\}$}{
            Add $e$ to $\Eprime$. \tcp{This is just for the analysis.}
            \If{$u,v$ are not matched in $\M$}{
                add $e$ to $\M$ with weight $r_e$
            }
        }
    }
    \Return $\M$
\label{AlgoOn}
\end{algorithm}

\section{Edge Arrival in General Graphs}
Our solution for this problem, \cref{AlgoOn} outlines very well the idea we described above: 
in order to bridge the gap between the fixed order employed by greedy (which is a $2$-approximation) and the adversarial one of prophet inequalities (where greedily adding edges would be arbitrarily bad), we utilize the single sample $S=\{s_e|e\in E\}$ that we have from the product over all edge distributions. Recall that we assume the $D_e$ to be independent.

The algorithm computes an (offline) greedy matching $\Ms$
on the graph with edge weights $S$. Then, for each vertex $v$, it interprets the weight of the edge incident to it in the greedy solution as a price $p_v$ for the adversarial-order online algorithm: whenever an edge $e = (u,v)$ arrives online and both endpoints are free, it is added to the matching $\M$ if and only if $r_e >\max\{p_v, p_u\}$\footnote{We assume to break ties uniformly at random before running the algorithm, i.e. for the set of all drawn numbers that have value $x$, pick a permutation $\pi$ of those uniformly at random and define a draw of the same value to be larger than another if and only if it comes first in $\pi$.}.
While the algorithm is simple, the analysis is a bit more challenging.
As pointed out before, the greedy method can be seen as an offline algorithm as well as an online algorithm with the (extremely simplifying) assumption of a fixed and weight-decreasing arrival order. We combine such a shift of viewpoint with a suitable reformulation of the random processes involved. More precisely, we look at our (online) algorithm as the byproduct of a run of a (fixed-order!) greedy strategy on a random part of a fictitious graph that contains each edge \emph{twice} (once with the true weight $r_e$, and once with the sampled weight $s_e$).

To help with the analysis, together with the solution $\M$, the algorithm builds $\Eprime$, a subset of $E$ containing all the edges such that $r_e > \max\{p_v, p_u\}$, i.e. all edges that are price-feasible. Note that $\Eprime$ is a superset of the actual solution $\M$ and may not be a matching.

\paragraph{Equivalent offline algorithm.} In order to analyze the competitive ratio of \cref{AlgoOn}, we state an offline, greedy-based algorithm with similar properties that is easier to work with.
This offline procedure considers an equivalent stochastic process generating the edge weights in $S$ and $R$. Instead of first drawing all the samples, and then all the realized weights, we consider a first stage in which two realizations of each random variable $D_e$ are drawn. Then, they are sorted in decreasing order according to their values, breaking ties at random as explained above.
Finally, the offline algorithm goes through all these values in that order; each value is then associated with $S$ or $R$ by a random toss of an unbiased coin. Consequently, the second, smaller weight of the same edge arriving later is associated to the remaining category with probability $1$. It is not hard to see that the distributions of $S$ and $R$ are identical for both random processes.
The pseudocode for this offline procedure is given in \Cref{AlgoOff}. As already mentioned, it mimics in an offline way the construction of $\Ms$ and $\Eprime$ of \Cref{AlgoOn}, as formalized in the following Claim.

\begin{algorithm}[t]
\DontPrintSemicolon
\caption{Offline Matching}
\label{AlgoOff}
    Set $\Eprime=\emptyset$, $\Ms= \emptyset$, $\M = \emptyset$ $V^S=V$ \;
    For each $e\in E$, draw from $D_e$ two values $a_{e,1}$ and $a_{e,2}$\;
    Order $A=\{a_{e,1}, a_{e,2}|e\in E\}$ in decreasing fashion \;
    Mark each edge $e\in E$ as unused\;
    \For{each value $a\in A$ in the above order}{
        \If{$a$ corresponds to edge $e = (u,v)$ s.t. $e$ is unused, $u\in V^S$ and $v\in V^S$ }{
        {Flip a coin \;
            \eIf{Heads}
            {
                Mark $e$ as \emph{$R$-used}\;
                Add $e$ to $\Eprime$ with  weight $a$\;
            }
            {
                Mark $e$ as \emph{$S$-used}\;
                Add $e$ to $\Ms$ with  weight $a$\;
                Remove $u, v$ from $V^S$  \tcp{Blocking these vertices for use in $M_S$ and $\Eprime$.}
            }
        }
        }
        \ElseIf{$a$ corresponds to $e = (u,v)$ that has been observed before, $u\in V^S$ and $v\in V^S$}{
            Remove $u$ and $v$ from $V^S$   \tcp{Blocking these vertices for use in $M_S$ and $\Eprime$.}
            Add $e$ to $\Ms$ with weight $a$
        }
    }
    \For{each $e$ in the same order as \Cref{AlgoOn}}{
        \If{$e$ in $\Eprime$ and $\{e\} \cup \M$ is a matching }
        {Add $e$ to $\M$ with weight $r_e$}}
    \Return $\M$.
\end{algorithm}
\begin{claim}
\label{claim:equivalence}
    The sets $\Eprime$, $\Ms$ and $\M$ are distributed in the same way when computed by \Cref{AlgoOn} as when computed by \Cref{AlgoOff}.
\end{claim}
\begin{proof}
    Consider any run of \Cref{AlgoOff}. Every edge value in $A$ is either associated to $R$ or to $S$, or does not belong to any set of interest. In fact, an edge value $a$ is associated to the $R$ graph if it is $R$ used or if the other value for the same edge gets $S$ used; the symmetric holds for $S.$ It may happen that an edge is never associated to either $R$ or $S$, but it means that one of its endpoints has been added to $M_S$, so the edge does not belong to $M_S$ (because $M_S$ is a matching) nor to $E'$ and $M$ (because it means that one of the endpoint is blocked by a large price). 
    Since the $S$ and $R$ values are all drawn independently from the product distribution on the $D_e$, for $e \in E$, they are distributed in the same way no matter if we draw them first and then decide which set they belong to (like in \Cref{AlgoOff}), or the other way around (like in \Cref{AlgoOn}).
    
    Once we have the equivalence of the underlying stochastic processes, we just need to show that for {\em any} realization of the $S$ and $R$ value, the corresponding sets of edges of interests are exactly the same in the offline and online settings. 
    To avoid complication, we deterministically use the same tie breaking rules in both the online and offline case. 
    
    Notice that if we restrict the offline algorithm to consider only the edge weights in the $S$ graph, we recover exactly the procedure to generate a greedy matching on the $S$ graph. This is true because what happens to the $R$ used edges, or to the second realization of $S$ used edges (i.e. all values associated with $R$) has no influence on $\Ms$. Hence, we can claim that the two versions of $\Ms$ in the different algorithms indeed follow the same distribution. Actually, this is even true step-wise: at any point in time during \Cref{AlgoOff}, $\Ms$ contains the greedy matching on the $S$ graph {\em restricted} to the edges whose $s_e$ is greater than the value being considered at that time. 
    
    We conclude the proof by showing that the two versions of $\Eprime$ coincide. Once we have that, it is enough to observe that the matchings $M$ are extracted in the same way from $\Eprime$. To this end, consider an iteration of \Cref{AlgoOn} in which an edge $e$ in the $R$ graph arrives.
    We argue that $e$ is added to $\Eprime$
    in the offline procedure if and only if the online procedure does the same. 
    
    Fix edge $e=(u,v)$ added to $E'$ by \Cref{AlgoOn}. Then, the online version of $\Ms$ contains no edges incident to $u$ or $v$ with larger weight than $e$, because this would imply a price higher than $r_e$. Since the $\Ms$ are the same in both algorithms, when \Cref{AlgoOff} considers $e$, it holds $u,\, v\in V^S$ and \Cref{AlgoOff} will also add $e$. The other direction is analogous.
\end{proof}
Given the above result, we will be using $\Eprime$, $\Ms$ and $\M$ mainly to refer to \cref{AlgoOff}, but the same would apply to \cref{AlgoOn}.

We are now ready for the main Theorem of the Section, in which we show that the competitive factor of \Cref{AlgoOn} is $16$, i.e. for any adversarial arrival order of the edges, in expectation the online algorithm retains at least a $1/16$ fraction of the value of optimal offline matching $OPT$. The high level idea of the proof lies in relating the optimal offline matching and our solution by using $\Ms$ (which yields a 2 approximation of $OPT$) and a carefully chosen subset of $\Eprime$, the {\em safe} edges. 

\begin{theorem}
\label{thm:edgeGen}
    For the problem of finding a maximum-weight matching in a general graph $G$, in the online edge arrival model, \cref{AlgoOn} is $16$-competitive in expectation, i.e. 
        \[
            16 \cdot \E{w(\M)} \ge \E{OPT}, 
        \]
        where $OPT$ is the weight of an optimal matching in $G$.
\end{theorem}
\begin{proof}
For any run of the offline algorithm, we say that an edge is {\em considered} if it is added to either $\Ms$ or $\Eprime$. Furthermore, let $V'$ be the subset of $V$ s.t. the according vertices all have at least one considered edge. Then, for each vertex $v$, we have the following inequality: 
\[
    \mathbb{P}[v\in E'| v\in V'] = 1/2.
\]
In fact, fix the point in the execution of the offline algorithm where the decision is made that $v\in V'$. This is the point in time the offline algorithm considers the first incident edge to $v$, implying that when we arrive at this value $a\in A$, both of the corresponding edges' endpoints are still free. Right after deciding to consider a value, a fair coin is flipped to determine whether the according edge is in $E'$ or $M_S$. 

Next, we claim the set $E'$ has expected weight at least $w(\Ms)/2$. This is true because for each considered edge, we can have one of two cases: either $e$ is added to both $\Ms$ and $\Eprime$, and in that case $r_e>s_e$, or $e \in  \Ms \setminus \Eprime$. In the latter case, the decision to which set $e$ is added was made by flipping a fair coin. This means, $\Ms$ and $\Eprime$-edges are determined in the exact same way, except that after already adding an $R$-version of $e$ to $\Eprime$, the $S$-version might still be used for $\Ms$.

Now define $val_{\hat{E}}(v)$ for any $v\in V$ and $\hat{E}\subseteq E$ to denote the largest weight of an edge incident to vertex $v$ in edge set $\hat{E}$, or $0$ if there is no such edge.
Then, it holds 
\begin{equation}
\label{eq:val_e}
    \mathbb{E}\left[\sum_{v\in V} val_{\Eprime}(v)\right]\geq \frac 12 \mathbb{E}\left[ \sum_{v\in V}val_{M_S}(v)\right]= \mathbb{E}\left[w(M_S) \right].
\end{equation}

We define $e_v$ as the edge (if any) that defines the value $val_{E'}(v)$, which is always the first edge incident to $v$ that was considered by the offline algorithm. Slightly abusing the notation, we define $r_{e_v} = 0$ if the edge $e_v$ does not exist.
We rewrite \cref{eq:val_e} as 
\begin{equation}
    \mathbb{E}\left[\sum_{v\in V} r_{e_v}\right]\geq \mathbb{E}\left[ w(M_S) \right],
\label{eq:r_ev}    
\end{equation}
note that some edges might occur twice in above sum.
Fix now some edge $e_v$, we call $e_v=(u,v)$ \emph{safe} if for $v$, $e_v$ is also the only edge in $E'$ and for $u$, there is no smaller incident edge in $E'$.
Focus on node $v$. $e_v$ is the first edge in $E'$ by definition, and it is the last if the next considered edge incident to $v$ is associated to $S$. This happens with probability at most $\frac 12$  since either this is the first copy of the edge seen, or it must be in $S$ anyway.

The latter happens if and only if the next edge after $e_v$ incident to $u$ in the offline algorithm is decided to be in $S$, which is also the case with probability at least $1/2$ for the same reasons as before.
Both events do depend on each other - one might influence the other in terms of which actual edge is the considered one or at which point in time the consideration takes place. But for the $R$- $S$- decision on the next incident edge, it holds that it is either made only at the point in time that edge is considered (via a fair coin flip), or it has been made before (i.e., edge is $R-$used). In the latter case, however, the second copy must be in $S$. 
This means, the probability for the next considered edge to be in $S$ can never be manipulated to be below $1/2$.

We showed that with probability at least $1/4$, any edge $e_v$ is safe. Fix now any realization of the edge weights in A, and consider only the randomness of the
coin tosses. Then,
\begin{align*}
    \mathbb{E}[\sum_{v\in V} r_{e_v} \mathbbm{1}_{\{e_v\text{ is safe}\}}] 
    &= \sum_{v\in V}\sum_{e\in E} r_e \cdot
        \mathbb{P}[e_v \text{ is safe}|e_v=e] \cdot \mathbb{P}[e_v=e]\\
    &\geq \sum_{v\in V}\sum_{e\in E} 
        \frac {r_e}4 \cdot \mathbb{P}[e_v=e] 
        = \frac 14 \E{\sum_{v \in V} r_{e_v}}.
\end{align*}

Finally, we need to take the set of $e_v$-edges and consider which subset of them will be part of the matching chosen by the prophet inequality. 
Note that in the \emph{safe} set, the only edge incident to vertex $v$ that might be chosen in the online algorithm is $e_v=(u,v)$, by definition. 
Now either, this indeed happens in the algorithm, or at the time $e_v$ arrives, vertex $u$ is already matched via some edge $e_u'$, which must be of larger weight than $e_v$ due to the definition of \emph{safe}.
Now, this also implies that $e_u$ is not safe, since apparently more than one edge is incident to $u$ in $E'$. 
All in all, we have that for each vertex $v$ s.t. $e_v$ is safe, either $v$ itself or a unique vertex $u$ with $e_u$ not safe will be matched via an edge of weight at least $r_{e_v}$.
Since in case we indeed have that a safe edge $e_v$ is added to $M$, $e_v$ can possibly also be the safe edge for its other endpoint (and our sum thus counts the weight of this edge twice), the actual matching computed by the algorithm recovers at least half of the summed-up values $val_{E'}(v)$ of the safe matching:
\[
\mathbb{E}\left[ w(M) \right]
    \geq 
    \frac 12 \mathbb{E}[\sum_{v\in V} r_{e_v} \mathbbm{1}_{\{e_v\text{ is safe}\}}] \ge \frac{1}{8} \E{w(\Ms)}
    \ge \frac{1}{16}\E{OPT}.  \qedhere
\]
\end{proof}

\section{Vertex Arrival in Bipartite Graphs}

    The main difference in the bipartite setting is, as noted in the introduction, that instead of having individual edge arrivals, one side of the graph is fixed and at each step a vertex arrives with all of its incident edge weights revealed. We refer to the fixed side as the `items' and the online side as the `buyers', even though in this section there is no incentive analysis. A truthful version of these results, with slightly worse approximation guarantees, is presented in \cref{sec:truthful}.
    
    As before, the sample $S$ is used to calculate prices, on both sides: items (fixed side vertices) have prices to \emph{protect} them from being sold too cheaply, while buyers (incoming vertices) have prices they need to beat in order to participate in the market, in addition to beating the item prices.

\begin{algorithm}
\DontPrintSemicolon
    Set $\Eprime=\emptyset$\;
    Compute the greedy matching $\Ms$ on the graph with edge weights according to sample $S$\;
    \For{all $e = (i,j) \in \Ms$}{
        Set $p_{i}=p_{j}=s_e$.
        }
    \For{all vertices $k$ not matched in $\Ms$}{
        Set $p_{k}=0$.
        }  
    \For{each arriving $i^* \in I$}{
        Let $\hat e = \argmax \{ r_e| \ e = (i^*,j)\text{ and } r_e > max\{p_{i^*},p_j\}\}$\;
        Add $\hat e$ to $\Eprime$, with weight $r_{\hat e}$.     \tcp{Only used for the analysis, just as before.}
        Let $j^*$ be the item in $\hat e$\;
        \If{$j^*$ is still free in $\M$}{
            Add $(i^*,j^*)$ to $\M$ with weight $r_{\hat e}$
        }
    }
    \Return $\M$
\caption{Bipartite Prophet Matching}
\label{AlgoBipOn}
\end{algorithm}

    We provide some intuition on the objects of interest that appear in the pseudo-code of Algorithm \ref{AlgoBipOn}. $M_S$ is the greedy matching on the $S$ graph, and it does not depend in any way on the actual edge weights $R$. $\Eprime$ is a subset of the edges of the $R$ graph: each buyer is associated with at most one item, but there may be items associated to more than one buyer. Solving these conflicts in $\Eprime$ in favor of the first arriving edge we obtain $\M$. 
    Note that $\M$ is computed in an online fashion.
  
    In order to relate $\Eprime$ and $\M$ with $OPT$, similarly to the previous section, we consider an offline version of the algorithm, i.e. \cref{AlgoBipOff} which exhibits the same distribution over the relevant allocations and can be analyzed more easily.

\paragraph{Equivalent offline algorithm.} The offline algorithm interleaves the building of a greedy matching $\Ms$ with that of the (possibly non feasible) allocation $\Eprime$. All the values are drawn in advance, but they are assigned to $S$ and $R$ only upon arrival.

\begin{algorithm}
\DontPrintSemicolon
\caption{Bipartite Offline Matching}
\label{AlgoBipOff}
    Set $\Eprime=\emptyset$, $\Ms= \emptyset$, $\M=\emptyset$\;
    $I^R=I^S=I$, $J^S=J$ \;
    For each $e\in I\times J$, draw from $D_e$ two values $a_{e,1}$ and $a_{e,2}$\;
    Order $A=\{a_{e,1}, a_{e,2}|e\in E\}$ in a decreasing fashion \;
    Mark each edge $e\in E$ as unused\;
    \For{each value $a$ from $A$ in above order}{
        \If{$a$ corresponds to edge $e$ s.t. $e$ is unused}{
        Flip a coin\;
        \eIf{Heads}{
            Mark $e$ as $R$-used
        }
        {Mark $e$ as $S$-used}
        }
        \If{$e$ has just been marked \emph{$R$-used}, or was previously \emph{$S$-used}, $i\in I^R$ and $j\in J^S$}{
            Add $e$ to $\Eprime$ with weight $a$ and remove $i$ from $I^R$.
            }
        \If{$e$ has just been marked \emph{$S$-used}, or was previously \emph{$R$-used}, $i\in I^S$ and $j\in J^S$}{
        Remove $j$ from $J^S$, and remove $i$ from $I^S$ and from $I^R$\;
            Add $e$ to $\Ms$ with weight $a$
        }
    }
    \For{each $e$ in the same order as \Cref{AlgoOn}}{
        \If{$e$ in $\Eprime$ and $\{e\} \cup \M$ is a matching }
        {Add $e$ to $\M$ with weight $r_e$}}
    \Return $\M$.
\end{algorithm}

As before, we observe that the allocations $\Eprime$ and $\M$ are distributed in exactly the same way, whether they are computed by \cref{AlgoBipOff} or \cref{AlgoBipOn}.
    
\begin{claim}
    The sets $\Eprime$, $\Ms$ and $\M$ are distributed the same way when computed by \Cref{AlgoBipOn} as when computed by \Cref{AlgoBipOff}.
\end{claim}
\begin{proof} 
    The proof is analogous to the one of \Cref{claim:equivalence}, the only difference being that here we assume vertex arrival, i.e. all edges incident to the same buyer arrive at the same time. The $R$ and $S$ graphs in the offline and online setting follows the same distribution as argued in the previous Section, and $\Ms$ is still the greedy matching on the $S$ graph, in both \Cref{AlgoBipOn} and \Cref{AlgoBipOn}. We know that $\M$ is extracted in the same way by $\Eprime$ by the two algorithms, so we just need to argue that for each fixed realization of the edge weights, the online and offline versions of the edge sets $\Eprime$ coincide.
    Considering now which edges $e=(i,j)$ from the $R$ graph are successfully added to $M'$, those are exactly the ones such that for both of their endpoints, no incident edge had an $S$-value greater than $r_e$ and no edge containing $i$ with the same property but larger value exists. This is exactly the set of edges fulfilling the argmax in \cref{AlgoBipOn}.
\end{proof}

    We are now ready to show the main result of this Section.
    \begin{theorem}
    \label{thm:vertexBip}
        For the problem of max-weight matching in a bipartite graph $G$, where vertices on one side arrive online, \cref{AlgoBipOn} is $8$-competitive in expectation, i.e. 
        \[
            8 \cdot \E{w(\M)} \ge OPT, 
        \]
        where $OPT$ denotes the weight of an optimal matching in $G$.
    \end{theorem}
    \begin{proof}

We start noting that for the greedy matching $M_S$ it holds that $\mathbb{E}[w(M_S)]\geq \frac 12 \mathbb{E}[OPT]$. Further, for the same reasons as in the edge arrival case, we have $\mathbb{E}[w(E')]\geq \frac 12 \mathbb{E}[w(M_S)]$. 

Especially, recall that for any vertex $v$, $e_v$ denotes the $v$-incident edge in $E'$ with largest weight, and $val_{E'}(v)=r_{e_v}$. Again, as in \cref{eq:r_ev}, we have
\[ \mathbb{E}\left[ \sum_{v\in V}r_{e_v} \right] \geq \mathbb{E}\left[w( M_S)\right] .\]

The definition of \emph{safe} edges, however, changes slightly: since now, all edges incident to some buyer $i$ arrive at the same time, for each $i$, there exists only at most one $i$-incident edge in $E'$ since after adding one, we erase $i$ from the available set.
Therefore, we define $e_v=(u,v)$ as \emph{safe} if for the item vertex in $j=\{u,v\}\cap J$, there exists no smaller $j$-incident edge than $e_v$ in $E'$. We note that any edge $e_v$ is \emph{safe} with probability at least $\frac{1}{2}$: after adding $e_v$ incident to $j$, when the next $j$-incident edge is considered in the offline algorithm, it is assigned to $S$ and therefore $M_S$ with probability at least $\frac 12$. This yields
\begin{align*}
    \mathbb{E}[\sum_{v\in V} r_{e_v} \mathbbm{1}_{\{e_v\text{ is safe}\}}] 
    &= \sum_{v\in V}\sum_{e\in E} r_e \cdot
        \mathbb{P}[e_v \text{ is safe}|e_v=e] \cdot \mathbb{P}[e_v=e]\\
    &\geq \sum_{v\in V}\sum_{e\in E} 
        \frac {r_e}2 \cdot \mathbb{P}[e_v=e] 
        = \frac 12 \E{\sum_{v \in V} r_{e_v}}.
\end{align*}
Consider an edge $e=(i,j)$ in the \emph{safe} set and note that at least one of its endpoints will be matched via an edge with value at least $r_e$ in $M$.
We account for this possible double-counting (since $e$ contributes to above sum for both endpoints) by dividing by $2$. Taking now the expectation over all the algorithm and the random values in $A$, we can conclude that
\[ \mathbb{E}\left[ w(M) \right]
    \geq 
    \frac 12 \mathbb{E}[\sum_{v\in V} r_{e_v} \mathbbm{1}_{\{e_v\text{ is safe}\}}] \ge \frac{1}{4} \E{w(\Ms)}
    \ge \frac{1}{8}\E{OPT} . 
    \qedhere
\]
    \end{proof}

\section{Truthful Bipartite Matching}\label{sec:truthful}
For the same model as in the previous section, with buyers on the left side of $G$ arriving online and items on the right side available offline, we design an incentive-compatible version of our prophet inequality.
The difficulty is that our $8$-approximation in \Cref{thm:vertexBip} relies on each buyer choosing the \emph{largest} incident edge that beats an item's price. However, if the price is close to the buyer's value, this edge might have very low utility. The buyer might prefer to purchase an item she has lower value for instead, but where the price is sufficiently lower, weakening our analysis to the point that we can no longer show any approximation.

Clearly, a different approach is in order. We use a pricing-based algorithm that does not rely on restricting the buyer's choice of item (among those with feasible prices that are still available). 
Instead, we ensure a sufficient weight for the resulting matching solely via appropriate thresholds for each buyer and item: this can be seen as doing a worst-case analysis over all threshold-feasible edges instead of choosing a good subset of them, explaining the incurred loss in the approximation. This allows us to make another small change and, among all threshold-feasible items, always assign the \emph{utility-maximizing} one.
In turn, the routine becomes truthful, and we have to mainly care about showing the approximation.

Our strategy for setting the aforementioned thresholds can be interpreted as follows.
Assume we take as a starting point not our prophet inequality for the bipartite case, but the one for general graphs with edge arrivals. 
Obviously, arrival orders for the bipartite one-sided model we consider here are highly restricted compared to general edge-arrival. However, assuming the graph is bipartite, each arriving edge will consist of exactly one buyer $i$ and one item $j$. Then, \cref{AlgoOn} assigns prices $p_i$ and $p_j$, where $p_i$ corresponds to the aforementioned buyer threshold we are planning to use.

Our mechanism works as follows: we determine greedy prices on the sample graph for all vertices just as described above. When a buyer $i$ arrives, we present her not with the greedy price of each item, but instead charge for each item $j$ the maximum over $p_j$ and $p_i$. Then, we let the buyer choose what she likes best, i.e. we assign a utility-maximizing item.

\begin{algorithm}
\caption{Truthful Bipartite}
\label{AlgoT}
\DontPrintSemicolon
Set $M_T=\emptyset$\;
Compute the greedy matching $\Ms$ on the graph with edge weights according to sample $S$. 
\For{all $e = (i,j) \in \Ms$}{
    Set $p_{i}=p_{j}=s_e$.
    }
\For{all vertices $k$ not matched in $\Ms$}{
    Set $p_{k}=0$.
    }  
\For{each arriving $i^* \in I$}{
        Set $F(i^*) = \{e = (i^*,j) | r_e > \max\{p_{i^*},p_j\} \text{ and } j \text{ is not matched in } M_T\}$\;
        Set $\tilde{e} = \argmax \{ r_e - max\{p_{i^*},p_j\}| \ e \in F(i^*)\}$\;
        Add $\tilde{e} = (i^*,j^*)$ to $M_T$ with weight $r_{\tilde{e}}$\;
        Charge buyer $i^*$ a price of $\max\{p_{i^*},p_{j^*}\}$
    }
    \Return $M_T$
\end{algorithm}
Here, $M_S$ is again a greedy matching on the sample graph.
As before (we omit this part of the analysis because it is analogous), one can imagine an offline version of \Cref{AlgoT} that simply draws two values for each edge, and then goes through their set $A$ in non-increasing order, deciding for each value considered whether it belongs to $R$ or $S$ (if this was not yet decided by the twin value for the same edge).

Also as before, the offline version will produce the same set of price-beating or \emph{feasible} edges $\Eprime$, which are the exact ones that \Cref{AlgoT} will consider (and add in case they are the buyer's favorite and the item is still free).

\begin{theorem}\label{thm:truthfulBip}
For the bipartite max-weight matching problem with one-sided vertex arrival, Algorithm \ref{AlgoT} is a truthful $16$-approximation, in expectation.
\end{theorem}
\begin{proof}
Truthfulness is immediate: on arrival, the unmatched items are offered to the buyer for prices determined independently of the buyer's values, and she is assigned one that is maximizing her utility. By misreporting, her utility can therefore only become worse.
Now for the approximation ratio, consider the edge set $\Eprime$ defined via the  offline version of our edge-arrival \cref{AlgoOff}. Note that the arrival order, be it edge-arrival or vertex-arrival, has no influence whatsoever on the computation of set $\Eprime$ in this algorithm.
With the same arguments as in the analysis of the edge-arrival model, $\Eprime$ from the offline version is distributed exactly as the set of all price-feasible edges $E'$ in \cref{AlgoT}.
We base our proof on the observation that after establishing $\mathbb{E}[w(\Eprime)]\geq \frac 12 \mathbb{E}[w(M_S)]$, all we do in the analysis of \cref{AlgoOn} is proving a lower bound on the weight of a maximal matching in edge set $E'$.
Concretely, since $\Eprime$ is the same as for edge arrival as well as vertex arrival in the offline versions of our algorithms, all inequalities hold as before. In addition, when a buyer arrives that beats both thresholds of at least one incident edge, she also has nonnegative utility for the edge when paying the maximum over both. Therefore, if an edge in $E'$ is incident to the buyer and the item is free at her arrival, she will match to some item (not necessarily this one, due to choosing the utility-maximizing option, but some).
Therefore, no buyer $i\in I$ s.t. $\exists j\in J\text{ s.t. } e=(i,j)\in E'$ will ever stay unmatched unless all according items are matched already. On the other hand, fixing an item $j$ with at least one incident edge in $\Eprime$, $j$ will either be matched or the according buyer chooses to match to a different item. Together, this implies that \cref{AlgoT} chooses indeed a maximal matching in edge set $\Eprime$.
Now consider again the edges $\{e_v|v\in V\}$ defined as before for edge arrival (and actually, computed as before, just on a subset of possible instances containing only the bipartite graphs). Recall that for every vertex $v$, $e_v$ is the largest-weight edge in $\Eprime$ incident to $v$, if any. 
For each edge $e_v=(i,j)$ that is also \emph{safe}, i.e., $E'$ contains no lower-weight edges incident to $i$ or $j$, neither $i$ nor $j$ can be matched in $M_T$ via a smaller edge than $e$.
At the same time, since $M_T$ is maximal in $\Eprime$, either $i$ or $j$ must be matched, yielding that $M_T$ recovers at least half of the summed-up values $\sum_{v\in V}r_{e_v}\mathbbm{1}_{e_v\text{ is safe}}$ analogously to the edge-arrival case.
\end{proof}

\section{Conclusion}
The single-sample paradigm is proving more and more powerful, even for combinatorially interesting problems like max-weight matching. We have designed constant-approximation prophet inequalities for a number of settings, and our results hold without restrictive assumptions like limited degree. However, approximation ratios clearly seem to depend on the generality of arrival model, independence assumptions on the distributions, the assumed structure of the underlying graphs, and whether or not one assumes the presence of selfish agents. In the future, it will be an interesting endeavor to fully understand these relationships of model and ratio, and derive tighter approximations as well as lower bounds. Finally, a range of other combinatorial problems for which prophet inequalities with small competitive ratio are known to exist in the full information model, are left to explore in the single-sample setup.

\section*{Acknowledgement}
The authors would like to thank David Naori, Danny Raz and Haim Kaplan for pointing out a mistake in a previous version of one of our proofs.

\bibliography{references}

\end{document}